\documentclass[%
reprint,
superscriptaddress,
 amsmath,amssymb,
 aps,
pra,
]{revtex4-2}

\usepackage{graphicx,color}
\usepackage{dcolumn}
\usepackage{bm}

\usepackage{graphicx,color}
\usepackage{textcomp}
\usepackage{amsthm}
\usepackage{mathtools}
\usepackage{algorithm, algpseudocode}
\usepackage{enumitem}
\usepackage{physics}
\usepackage[export]{adjustbox}
\usepackage{appendix}
\usepackage[normalem]{ulem}
\usepackage{amsmath, amssymb, amsthm}
\usepackage{tikz}
\usetikzlibrary{matrix}
\mathtoolsset{showonlyrefs=true}

\usepackage{etoolbox}

\AtEndEnvironment{theorem}{\hfill$\square$}

\algrenewcommand\algorithmicindent{1em}

\theoremstyle{remark}

\usepackage[dvipsnames]{xcolor}

\newtheoremstyle{boldthm} 
  {3pt}                   
  {3pt}                   
  {\itshape}              
  {}                      
  {\bfseries}             
  {.}                     
  { }                     
  {}                  
  
\theoremstyle{boldthm}
\newtheorem{theorem}{Theorem}[]
\newtheorem{lemma}{Lemma}[]
\newtheorem{proposition}{Proposition}[]
\newtheorem{definition}{Definition}[]

\newcommand{\EX}{\mathbf{E}}
\newcommand{\PR}{\mathbb{P}}

\newcommand{\oppar}[1]{\left( #1 \right)}

\def\mc{\mathcal}

\usepackage{xr}
\makeatletter
\newcommand*{\addFileDependency}[1]{
\typeout{(supplementary)}
\@addtofilelist{supplementary}
\IfFileExists{supplementary}{}{\typeout{No file supplementary.}}
}\makeatother

\newcommand*{\myexternaldocument}[1]{%
\externaldocument{supplementary}%
\addFileDependency{supplementary.tex}%
\addFileDependency{supplementary.aux}%
}
\myexternaldocument{supplementary}

\usepackage{xstring}
\newcommand{\secref}[1]{\IfBeginWith{#1}{sec:}{S-}{}\ref{#1}}

\begin{document}

\preprint{APS/123-QED}

\title{\textbf{Expander qLDPC Codes against Long-range Correlated Errors in Memory} 
}%

\author{Yash Deepak Kashtikar}
\email{Student authors are listed in alphabetical order.}
\affiliation{%
 Department of Electrical Engineering, Indian Institute of Technology Madras,  India
}
\author{Pranay Mathur}
\email{Student authors are listed in alphabetical order.}
\affiliation{%
 Department of Electrical Engineering, Indian Institute of Technology Madras,  India
}
\author{Sudharsan Senthil}
\email{Student authors are listed in alphabetical order.}
\affiliation{%
 Department of Electrical Engineering, Indian Institute of Technology Madras,  India
}
\author{Avhishek Chatterjee}%
 \email{Contact author: avhishek@ee.iitm.ac.in}
\affiliation{%
 Department of Electrical Engineering, Indian Institute of Technology Madras, India
}%

\date{\today}

\begin{abstract}
Fault-tolerance using constant space-overhead  against long-range correlated errors is an important practical question. In the pioneering works \cite{TerhalB2005,AliferisGP2005,AharonovKP2006}, fault-tolerance using poly-logarithmic overhead against long-range correlation modeled by pairwise joint Hamiltonian  was proven when the total correlation of an error at a qubit location with errors at other locations was  $O(1)$, i.e., the total correlation at a location did not scale with the number of qubits. This condition, under spatial symmetry, can simply be stated as the  correlation between locations decaying faster than $\frac{1}{\text{dist}^{\text{dim}}}$. However, the pairwise Hamiltonian model remained intractable for constant overhead codes. Recently, \cite{BagewadiC2024memory} introduced and analyzed the generalized hidden Markov random field (MRF) model, which provably captures all stationary distributions, including long-range correlations \cite{KunschGK1995}. It resulted in a noise threshold in the case of long-range correlation, for memory corrected by the linear-distance Tanner codes \cite{LeverrierZ2022} for super-polynomial time. In this paper, we prove a similar result for square-root distance qLDPC codes and provide an explicit expression for the noise threshold in terms of the code rate, for up to $o(\sqrt{\text{\#qubits}})$ scaling of the total correlation of error at a location with errors at other locations.

\end{abstract}

\maketitle

Quantum low density parity check codes (qLDPC) codes \cite{TillichZ2013qLDPC,KovalevP2013LDPC,gottesman2014fault,LeverrierTZ2015qexpander} are known to have a polynomial distance and constant  space overhead. This is a significant improvement over the poly-logarithmic space overhead required for surface codes, and may lead to practically useful quantum computers using sub-million physical qubits. Expander qLDPC codes with sublinear distance \cite{TillichZ2013qLDPC} were shown to provide fault-tolerance against i.i.d. and local stochastic errors  \cite{KovalevP2013LDPC,gottesman2014fault,FawziGL_STOC2018,Grospellierthesis}. Furthermore, sublinear-distance qLDPC codes have shown encouraging performance in moderately sized practical quantum memories.\cite{BravyiCG2024high}.

Despite significant advances in theory and practice towards achieving fault-tolerance, some criticisms remain. Two main issues raised by critics are: (i) whether syndromes can be extracted from constant rate codes using (topologically) local gate operations, and (ii) whether the constant overhead fault-tolerance schemes can correct long-range correlated errors. Recently, some interesting work has been done to address the first issue, where promising local and semi-local syndrome extraction methods have been proposed \cite{BerthusenGG2025,PattisonKP2025}. In this work, we try to address the second issue, that of fault-tolerance against long-range correlated errors.

Fault-tolerance against the local stochastic error model is an important step towards studying correlation. However, since this model has an exponential correlation decay \cite{TerhalB2005}, the problem of fault-tolerance against long-range correlation remains open. The seminal series of papers \cite{TerhalB2005,AliferisGP2005,AharonovKP2006} showed fault-tolerance against long-range correlated errors using poly-logarithmic space overhead. 

Though a well-chosen system and bath joint Hamiltonian can model any correlation, it is often analytically intractable. In \cite{TerhalB2005,AliferisGP2005,AharonovKP2006},  pairwise joint Hamiltonian models were introduced and it was shown that if for any qubit its total correlation with all other qubits is $O(1)$, fault-tolerance can be achieved. Under spatial symmetry, this condition  is satisfied if the correlation between two locations decays faster than $\frac{1}{\text{dist}^{\text{dim}}}$.

However, studying constant overhead codes such as qLDPC codes, using the pairwise joint Hamiltonian model remained intractable. Recently, in \cite{BagewadiC2024memory}, the generalized hidden Markov random field model was proposed, which could model all stationary distributions including the ones where any qubit's total correlation with other locations scale as $({\text{\#qubits}})^{\gamma}$, for  $0\le \gamma \le 1$. This model was proven to capture correlation structures that are not captured by pairwise joint Hamiltonian models \cite{BagewadiC2024memory}. Also, any quantum memory error corrected using linear  distance Tanner codes (from \cite{LeverrierZ2022,GuPT2023}) was shown to have a super-polynomial retention time when the total correlation of a qubit location with others grows no faster than $\sqrt{{\text{\#qubits}}}$.

In this paper, we study fault-tolerance of the square-root distance expander qLDPC codes   \cite{TillichZ2013qLDPC,KovalevP2013LDPC,gottesman2014fault,LeverrierTZ2015qexpander} in long-range correlated errors modeled as a generalized hidden Markov random field (MRF)  \cite{BagewadiC2024memory} and show a positive noise threshold (of $\Theta(1)$) when the total correlation of a qubit grows no faster than $\sqrt{{\text{\#qubits}}}$. To our knowledge, this is the first provable noise threshold result for sublinear distance constant overhead codes against long-range correlations where the total correlation at a location increases with the number of qubits. 

\subsection*{Fault-tolerant memory}

The memory contains a state of $n$ qubits that encodes $k$ logical qubits using a stabilizer code. The evolution of this memory is modeled as a sequence of periodic phases, each of which contains a \textit{rest phase} and an \textit{error correction phase}. The $n$-qubit state can decohere and accumulate errors during the rest phase. Since an arbitrary error can be decomposed as a linear combination of Pauli flips \cite{aharonov1997fault, AharonovB2008}, following the existing literature on fault-tolerance using qLDPC codes \cite{gottesman2014fault,fawzi2020constant, YamasakiK2024time}, we consider only Pauli errors.

The error correction phase has three main sub-phases. The first sub-phase is the extraction of syndrome bits by employing ancillary qubits. Next, based on the syndrome bits, a decoding algorithm estimates the locations and types of Pauli flips ($X$ or $Z$) that have occurred. Finally, appropriate Pauli operators are applied to those locations to correct the errors. Depending on the choice of code and decoding algorithm, and the distribution of Pauli errors, all qubit errors may or may not be corrected in an error correction phase.

We index the periodic phases by discrete time-steps $t=1, 2, \ldots$ and use the notation $\{L_{i,t}\}$ to denote the locations at which Pauli errors occur during the $t^{th}$ rest phase. An error at location $i \in [n]$ at time-step $t$ is indicated by setting $L_{i,t}=1$, otherwise $L_{i,t} \text{ is set to }0$. We use the notation $E_t \subset [n]$ to denote $\{i: L_{i,t}=1\}$. We call $E^e_t$ the effective error before the error correction phase $t$, which includes $E_t$ and the uncorrected errors from previous phases that have been carried forward.

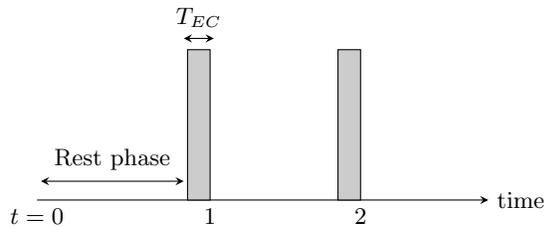
\begin{figure}
\label{fig:model}
\begin{tikzpicture}[>=stealth]

  \draw[->] (0,-1) -- (6,-1) node[right] {time};
  \draw (0,-1) node[below] {$t=0$};

  \def\tone{2}
  \def\ttwo{4}
  \draw[fill=gray!40] (\tone,-1) rectangle ++(0.3,2);
  \draw[fill=gray!40] (\ttwo,-1) rectangle ++(0.3,2);

  \draw (\tone+0.3,-1) node[below] {$1$};
  \draw (\ttwo+0.3,-1) node[below] {$2$};

  \draw[<->] (0.05,-0.75) -- (1.95,-0.75)
    node[midway,above=1pt] {Rest phase};
  \draw[<->] (1.95,1.15) -- (2.35,1.15)
    node[midway,above=1pt] {$T_{{EC}}$};
  
  \end{tikzpicture}
  \caption{Illustration of the quantum memory model. The labeled discrete times correspond to the end of a periodic phase. Each of these phases consists of a rest phase and an error correction phase. Qubits can decohere during the rest phase, it is hence  followed by a $T_{EC}$-long error-correction phase. The error-correction phase includes syndrome extraction as well as correction.}
  \end{figure}

\subsubsection*{Long-range Correlated Qubit Errors}

Long-range correlations between errors can be spatial and temporal. In phase $t$, there are correlations between $\{L_{i,t}: i \in [n]\}$, and there are correlations also between phases, i.e., between $\{L_{i,t}: i \in [n]\}$ and $\{L_{i,t'}: i \in [n]\}$ or, in other words, between $E_t$ and $E_{t'}$. 

For almost all systems in equilibrium with the environment, errors are stationary in time and space, i.e., $\PR(L_{i,t}=1)$ is the same for all $t$ and $i$ and we denote it by $\bar{p}$. It is known from the probability literature \cite{KunschGK1995} that a hidden Markov random field model can approximate any stationary distribution with arbitrarily high precision. Motivated by this fact and the analytical tractability of hidden random fields, a generalized hidden Markov random field model was proposed in \cite{BagewadiC2024memory} to model long-range correlated errors.  This error model was proven to be broader than the well known pairwise Hamiltonian model for long-range correlation  \cite{TerhalB2005,AliferisGP2005,AharonovKP2006}, and to include long-range correlations not captured by the said model. In this work, we study fault-tolerance using qLDPC under this generalized hidden Markov random field model.  

In the system, $\{L_{i,t}: i \in [n]\}$ are the manifest errors, which are caused by underlying (hidden) physical scenarios in the vicinity of the qubit locations. The physical scenario in the vicinity of a qubit is dictated by the local temperature, local circuit disturbances, local effects of external fields and radiations \cite{McEwenFA2022,Google2023}, etc. The hidden parameters at time $t$ are denoted by $\{J_{i,t}: i \in [n]\}$. Each $J_{i,t}$ is a vector of potentially multiple local physical parameters. We assume $J_{i,t}$ to be from a countable set, possibly infinite.

The probability of having an error at location $i$, $\PR[L_{i,t}=1]$, is potentially affected by multiple of these hidden parameters, and these hidden parameters are also highly correlated with each other. We use the following generalized hidden Markov random field model from \cite{BagewadiC2024memory} to capture long-range correlation.

\begin{align}
\PR(\{L_{i,t}: i \in [n]\}) & = \PR(\{J_{k,t}: k \in [n]\})  \nonumber \\
& ~~~~~~ \times \prod_i \PR(L_{i,t}|\{J_{k,t}: k \in [n]\}) \mbox{,   and   } \nonumber \\
\PR(\{J_{k,t}: k \in [n]\}) & = \PR(J_{1,t}) \prod_{k=1}^{n-1} \PR(J_{k+1,t}|J_{k,t}). \nonumber
\end{align}

For each $i \in [n]$, we define $g_i(\{J_{k,t}: k \in [n]\}):=\PR(L_{i,t}|\{J_{k,t}: k \in [n]\})$  such that $\EX[g_i(\{J_{k,t}: k \in [n]\})=\bar{p}$. We assume that each $g_i$ varies smoothly with changes in $\{J_{i,t}\}$. Formally, we assume $g_i$ to be $c_n$-Lipschitz with respect to the Hamming distance,  with $c_n =O(n^{-0.5-\epsilon_g})$ for  $0<\epsilon_g\le 0.5$. In \cite[Prop.~1]{BagewadiC2024memory}, an information theoretic converse was proven for a class of non-smooth $\{g_i\}$: despite the total correlation being $O(1)$ at any location, the memory retention time is $O(n^2)$. Hence, the assumption regarding the smoothness of $\{g_i\}$ seems unavoidable for achieving fault-tolerance over a long duration (high polynomial or super-polynomial).

Furthermore, the assumption on $\{g_i\}$ is a generalization of a condition that would likely arise in quantum chips when they scale. A major reason for the success of classical chips was that the area of the chips did not scale with the number of bits. This slow scaling of the area is also desired for quantum chips for them to be useful in practice. Such scaling would essentially require a linear number of qubits to be within a short physical distance, in turn causing the probability of error of a given qubit to be affected by the local environments of a linear number of qubits. However, $g_i$ being a measure of probability, can be $1$ at most. The individual impact of each of these $\Theta(n)$ neighboring locations would therefore have to be $\Theta(n^{-1})$. This scenario is equivalent to $g_i$  being $c_n$-Lipschitz with $c_n =O(n^{-1})$. Thus, the scenario we consider, i.e., $c_n=O(n^{-0.5-\epsilon_g})$ for $0<\epsilon_g\le 0.5$, is a more general, relaxed condition and would also be true advanced technologies in the future with improved insulation between qubits on a chip.

\subsubsection*{Metrics for Long-range Dependence}

For most physical systems, the physical parameters are stationary in a stochastic sense. Hence, it is natural to assume that $\{J_{i,t}\}$ are stationary in time, which in turn would imply that $\{L_{i,t}\}$ are stationary in time. However, we do not assume any temporal correlation decay or temporal Markovian property  for the hidden and manifest variables, $\{J_{i,t}\}$ and $\{L_{i,t}\}$.

Next, we discuss the notion of correlation in the study of long-range correlation. In the pioneering series of works on long-range correlation \cite{TerhalB2005,AliferisGP2005,AharonovKP2006}, the model involved pairwise terms $H_{ij}$, for qubit locations $i$ and $j$, in the system and bath joint Hamiltonian. The condition under which a positive fault-tolerance threshold was proved in \cite{AharonovKP2006} is: for each $i \in [n]$, $\sum_{j} ||H_{ij}||=O(1)$. 

For characterizing how the interactions or correlations between errors should decay with distance to have fault-tolerance ,this condition was simplified  to: $||H_{ij}||$ decaying faster than $|i-j|^{-\text{dim}}$. However, $\sum_{j} ||H_{ij}||=O(1)$ is a broader condition, since it can be true despite the decay of $||H_{ij}||$ being slower than $|i-j|^{-\text{dim}}$. 

To quantify the dependence between two random variables $Y_1$ and $Y_2$, a well-known metric in probability theory, computer science, and statistical physics \cite{BoucheronLM2013} is the total variation distance between their joint distribution and the product of their marginal distributions. Here, we denote that by $\text{Dep}(Y_1,Y_2)$. When this metric is close to $0$, i.e., $Y_1$ and $Y_2$ have low dependence, and $Y_1, Y_2$ take values in a finite set, their covariance, which we denote by $\text{cov}(Y_i,Y_j)$, is a Lipschitz continuous function with respect to this metric and hence can be bounded by a constant times the metric. Thus,  when $\{\text{Dep}(Y_i,Y_j)\}$ are close to $0$, the total accumulated covariance at $i$, $C(Y_i):=\sum_j \text{cov}(Y_i,Y_j)$ is  equivalent to $\text{Dep}(Y_i)=\sum_j \text{Dep}(Y_i,Y_j)$, the total dependence of $Y_i$ on $[n]$, in the order sense, i.e., $C(Y_i)=\Theta(\text{Dep}(Y_i))$.

The total number of errors $|E|$, also given by $\sum_i L_i$, has the mean $\bar{p}\cdot n=\Theta(n)$. Its variance, given by $\sum_i C(L_i)$, is $O(n)$, when $C(L_i)=O(1)$. It is implicit that fault-tolerance is not possible without having all but a vanishing fraction of $\{\text{Dep}(L_i,L_j)\}$ and $\{\text{cov}(L_i,L_j)\}$ close to $0$, otherwise the variance of $|E|$ would be $\Theta(n^2)$ and error correction would be impossible in an information theoretic sense.

In the pairwise joint Hamiltonian model, $||H_{ij}||$ captures the dependence between error at $i$ and $j$, and thus is equivalent to  $\text{Dep}(L_i,L_j)$. So, the quantity $\sum_{j} ||H_{ij}||$ is the total influence or correlation of all qubits on the qubit at location $i$ and is equivalent to the total dependence of qubit $i$ on $[n]$, $\text{Dep}(L_i)=\sum_j \text{Dep}(L_i,L_j)$, and the total accumulated covariance at $i$, $C(L_i)=\sum_j \text{cov}(L_i,L_j)$.  

 Thus, in this sense, the condition derived in \cite{TerhalB2005,AliferisGP2005,AharonovKP2006} for the dependence or correlation structure  to  achieve fault-tolerance can also be seen in three equivalent ways: (i) $\text{var}(|E|)$ is $O(n)$ i.e., $\text{var}(\sum_i L_i)$ is $O(n)$, (ii) $\text{Dep}(L_i)$ is $O(1)$, and (iii) $C(L_i)$ is $O(1)$.

In the generalized hidden MRF model \cite{BagewadiC2024memory}, the metric used to capture the total influence or correlation with all qubits on qubit $i$, denoted by $V(J_i)$, is defined as
$$\sum_j \sup_{a,b} \frac{1}{2} \sum_h |\PR(J_{j}=h|J_{i}=a) - \PR(J_{j}=h|J_{i}=b)|.$$
Each term in the summation corresponding to $j$ is equivalent to $\text{Dep}(J_i,J_j)$, the total variation distance between the joint and the product of the marginals of $J_i$ and $J_j$. Thus, $V(J_i)$ is equivalent to $\text{Dep}(J_i)=\sum_j \text{Dep}(J_i,J_j)$ in the order sense. 

When $\{g_i\}$ are smooth, fault-tolerance using linear distance Tanner codes is achieved when the total correlation at qubit $i$ with all other qubits, $V(J_i)$, is $o(n^{0.5})$  \cite{BagewadiC2024memory}. This is a significant  improvement over the previous condition on variance and total correlation in the pairwise Hamiltonian model, $C(L_i)=O(n)$, if $V(J_i)$ is equivalent to $\text{Dep}(L_i)$ and $C(L_i)$ in the order sense.

By the law of total covariance and the distribution of $\PR(\{L_{i,t}: i \in [n] \})$ defined above, the covariance between $L_i$ and $L_j$ is the same as the  covariance between $g_i(\{J_k: k \in [n]\})$ and $g_j(\{J_k: k \in [n]\})$. Consider the case where $g_i(\{J_k: k \in [n]\})=\sum_{k} c_{i,k} J_k$ for $c_{i,k} \ge 0$, $J_k \ge 0$ and $J_k$ belong to a finite alphabet. We need to have $\sum_{k} c_{i,k}=O(1)$ since $g_i$ are probabilities. Then, the covariance of $L_i$ and $L_j$ can be written as 
$$ \text{cov}(L_i, L_{j}) = \sum_{k,k'} \text{cov}(J_k, J_{k'}) c_{i,k} c_{j,k'}.$$

Thus, 
\begin{align*}
C(L_i) & =\sum_j \sum_{k,k'} \text{cov}(J_k, J_{k'}) c_{i,k} c_{j,k'} \\
& = \Theta(1) \cdot \sum_{k,k'} \text{cov}(J_k, J_{k'}) c_{i,k} \\
& = \Theta(1) \cdot \sum_k C(J_k) c_{i,k} =\Theta(C(J_i)),
\end{align*}
where the last step is due to spatial symmetry and the fact that $\sum_{k} c_{i,k}=O(1)$. Since $C(J_i)$ is equivalent to $\text{Dep}(J_i)$, which is equivalent to $V(J_i)$, we see that $C(L_i)$ and $V(L_i)$ are equivalent to $V(J_i)$ in the order sense. 

Hence, the condition under which fault-tolerance is proved for linear distance Tanner codes in \cite{BagewadiC2024memory} is indeed equivalent to $C(L_i)$ and $V(L_i)$ being $o(n^{0.5})$ and thus is an order-wise improvement over $C(L_i)$ and $V(L_i)$ being $O(1)$ in \cite{Terhal2015,AliferisGP2005,AharonovKP2006}. 

Since $V(J_i)$ is equivalent to $V(L_i)$ and $C(L_i)$, we use $V(J_i)$ as our metric for the total correlation of qubit $i$ with all qubits,  as in \cite{BagewadiC2024memory}, and prove fault-tolerance for  $V(J_i)=o(n^{\epsilon_g})$. As discussed above, in the likely case of large quantum chips, $\epsilon_g$ is $0.5$.

From a physical perspective, $\{g_i\}$ are continuous non-decreasing functions of $\sum_{k} c_{i,k} J_k$, like the sigmoid function. Since in fault-tolerance, we are interested in the case when probabilities of errors $g_i$ and $g_j$ are close to $0$ (low noise), in that range $g_i$ and $g_j$ can be lower and upper bounded by two lines through the origin. Hence, for small error rate, the above order-wise equivalence shown between $C(J_i)$ (or $V(J_i)$) and $C(L_i)$ (or $V(L_i)$) for linear $\{g_i\}$ also holds  for nonlinear $\{g_i\}$.

\subsection*{qLDPC for fault-tolerant memory}

Quantum Low Density Parity Check(qLDPC) codes are a well-studied group of stabilizer codes. Like stabilizer codes \cite{nielsen2002quantum}, they are constructed by combining two classical codes, one to correct Pauli $X$ flips and the other to correct Pauli $Z$ flips. The corresponding parity check matrices $H_X$ and $H_Z$ are such that the $Z$ code is the dual of the $X$ code. For constructing qLDPC codes, we need the two classical codes to be LDPC codes. The low density parity check property of a classical code means that its parity check matrix is sparse, where sparsity refers to the weights of the rows and columns being constant, i.e., not scaling with the size of the codeword. 

An effective way to construct classical low density parity check codes is to use expander codes \cite{LeverrierTZ2015qexpander,FawziGL_STOC2018}. Expander codes are constructed using bipartite expander graphs that serve as Tanner graphs for the code. The left nodes are code bits and the right nodes are parity checks. 

In this work, we use a particular class of quantum LDPC (qLDPC) codes, the quantum expander codes introduced in \cite{LeverrierTZ2015qexpander}. This code has a distance of $\Omega(\sqrt{n})$, as shown in \cite{FawziGL_STOC2018, Grospellierthesis}.  In \cite{Grospellierthesis}, these codes were shown to have a small set flip decoding algorithm with locality property. It was also shown to have fault-tolerance against errors distributed according to the local stochastic noise model that has a fast (exponential) correlation decay \cite{TerhalB2005}.  

In this paper, our main contribution is to show that the same codes can offer fault-tolerance against long-range correlated errors with significantly slower correlation decay  (polynomial with degree $1$) . Thus, this work gives an affirmative answer to whether constant overhead fault-tolerance is possible against long-range correlated errors. 

Our proof techniques extend to other qLDPC codes that have a polynomial distance and share a similar locality property of the decoding algorithm. There are recent qLDPC codes with these properties \cite{TamiyaKY2024polylog,YamasakiK2024time}.  However, in this paper, for the sake of simplicity in presentation, we limit our discussion to the $\Omega(\sqrt{n})$-distance expander qLDPC codes that were studied in \cite{LeverrierTZ2015qexpander,FawziGL_STOC2018,Grospellierthesis}.

As argued in previous work on fault-tolerance \cite{FawziGL_STOC2018,Grospellierthesis}, it is enough to prove fault-tolerance using expander qLDPC codes assuming only the possibility of Pauli-$X$ errors , as the case of Pauli-$Z$ errors follows similarly. Hence, for simplicity in presentation, the variables $\{L_{i,t}\}$, $E_t$ and $E_t^e$ can be assumed to correspond only to Pauli-$X$ errors. 

As discussed above, the fault-tolerant memory model considered here is standard: a state is stored after encoding it using the aforementioned expander qLDPC codes, and a polynomial time decoding algorithm is run in each error-correction phase. We employ the small-set-flip decoding algorithm \cite[Alg.~2]{FawziGL_STOC2018}.

\subsection*{Noise threshold against long-range correlated errors}

In this section, we present our main result on the existence of a noise threshold for memory affected by long-range correlated errors. In that regard, we first introduce the essential parameters that describe the expander qLDPC code, the decoding algorithms, and an appropriate metric for spatial correlation among errors.

The expander qLDPC code used here is the symmetric version of the code in \cite{FawziGL_STOC2018}, where the chosen bipartite expander graph has symmetric expansion coefficients  for the left and right sets of vertices, i.e., $\delta_A=\delta_B$  and $\gamma_A=\gamma_B$, as was the choice in \cite{Grospellierthesis}. However, the left and right degrees ($d_A$ and $d_B$, or $d_V$ and $d_C$, in \cite{FawziGL_STOC2018} and
\cite{Grospellierthesis} respectively) are different, and their ratio is denoted by $r$. The rate of the expander qLDPC code has a lower bound of $\frac{(1-r)^2}{1+r^2}$, which follows from \cite{TillichZ2013qLDPC,LeverrierTZ2015qexpander, FawziGL_STOC2018}.

As mentioned above, in each error correction phase, the linear-time small set flip decoding algorithm (Alg.~2 in \cite{FawziGL_STOC2018}) is run to periodically correct the errors that accumulate over time. The number of adversarial errors that can be corrected by the original small set flip decoding algorithm has a lower bound of $c' \sqrt{n}$, where \newline $c'=\frac{r (1-8\delta_A)}{4+2r(1-8\delta_A)} \gamma_A \frac{r^2}{\sqrt{1+r^2}}$, by Proposition~4.16 in \cite{Grospellierthesis}.  

\begin{definition}\label{def : Adjacency graph} Adjacency graph $\mathcal{G}$. \cite{FawziGL_STOC2018}.

The set of qubits is represented as a graph $\mathcal{G} = (\{1,\dots,n \}, \mathcal{E})$ called an \textit{adjacency graph}  with vertices representing the qubits of the code and edges representing the presence of a stabilizer generator acting on both qubits. 

\end{definition}

The graph $\mathcal{G}$ in our case is the same as  in \cite{FawziGL_STOC2018}. and has a $O(1)$ degree denoted by $d_\mathcal{G}$. 

\begin{proposition}
\label{prop:qLDPClongNoSynd}
Suppose a state is stored using a constant rate expander qLDPC code from \cite{LeverrierTZ2015qexpander,FawziGL_STOC2018,Grospellierthesis} with parameters $(d_A, d_B)$ and $\delta_A=\delta_B<\frac{1}{32}$, and it is periodically error corrected using the small-set-flip decoder \cite[Alg~2]{FawziGL_STOC2018}. 
Suppose that the memory has been subjected to the long-range correlated error discussed above where the average error rate per epoch is $\bar{p}$, $\{g_i: i \in [n]\}$ are $c_n$-Lipschitz with respect to the Hamming distance with $c_n=O\left( \frac{1}{n^{0.5+\epsilon_g}}\right)$, and the total correlation at $i$, $V(J_i)$, scales as $o(n^{\gamma})$ for $0\le \gamma<\epsilon_g\le 0.5$. Then, for $C_{\mc{G}}=\ln\left((d_{\mc{G}}-1) (1+\frac{1}{d_{\mc{G}}-2})^{d_{\mc{G}}-2}\right)$, there exists  $p_{th} = \exp\left(- \frac{808 \cdot (1+\frac{3 r}{8})\cdot C_{\mc{G}}}{300 \cdot r}\right)$ such that for all $\bar{p}<p_{th}$, if the state is retrieved at a time $T \le \exp(O(n^{\epsilon_g-\gamma}))$, it would have  fidelity $1-o(1)$.
\end{proposition}

For the same code parameters used in \cite{FawziGL_STOC2018} for numerical evaluations, the threshold value here is similar: $10^{-16}$--$10^{-15}$. Thus, the noise threshold above for long-range correlated errors is comparable to that of local stochastic and i.i.d. errors. Note that  $\gamma$ has to be less than $\epsilon_g$ in this proposition and $\epsilon_g$ bounded by $0.5$ to ensure $\bar{p}=\Omega(1)$ . Hence, the slowest correlation decay and the fastest growth of total correlation at a location, for which the proposition ensures a positive threshold for a super-polynomial memory retention time, are $\frac{1}{\sqrt{n}}$ and $\sqrt{n}$ respectively. 

The noise threshold in Proposition~\ref{prop:qLDPClongNoSynd} depends directly on the rate of the chosen code, which is known to have a lower bound of $\frac{(1-r)^2}{1+r^2}$, from \cite{TillichZ2013qLDPC,LeverrierTZ2015qexpander, FawziGL_STOC2018}. If this constant code rate is small, i.e., $r$ is close to $1$, \\ $p_{th} \gtrapprox \exp\left(- \frac{808 \cdot C_{\mc{G}}}{300}\right) ~\cdot~ \exp\left(- \frac{1.01 \cdot C_{\mc{G}}}{1-\sqrt{2 \cdot \text{code rate}}}\right) $. 

Since the code rate decreases with increasing $r$, the noise threshold $p_{th}$ decreases with higher code rates. This implies the existence of a trade-off between the noise threshold and the constant space overhead. This is not unexpected, since even for classical codes, codes with higher rates are less suited to correct more frequent errors.

\subsection*{Proof of Proposition~\ref{prop:qLDPClongNoSynd}}

To prove the main result in Proposition~\ref{prop:qLDPClongNoSynd}, we use the following general theorem, which relates the noise threshold to parameters of the chosen code and the constant-time decoding algorithm.

\begin{theorem}\label{thm: general Threshold} For a 1-dimensional arrangement of qubits, for any time $t \le \exp(O(n^{\epsilon_g-\gamma}))$, the memory state can be accurately retrieved at $t$ with probability $(1 - o(1))$ if   $\bar{p}<\exp(-\frac{101 \cdot C_{\mc{G}}}{100 \cdot \alpha})$ and  the long-range correlated error satisfies the conditions in Proposition~\ref{prop:qLDPClongNoSynd}, where $C_{\mc{G}}=\ln\left((d_{\mc{G}}-1) (1+\frac{1}{d_{\mc{G}}-2})^{d_{\mc{G}}-2}\right)$.  Here 
$\alpha=\frac{\frac{r}{2}(1-8\delta_A)}{\frac{r}{2}(1-8\delta_A)+1}$.

\end{theorem}

This theorem states that using the small-set-flip decoder in each error correction phase, the probability of accurate retrieval asymptotically reaches 1 when the error probability is below the given threshold.
Next, we use this result to prove the main result in Proposition~\ref{prop:qLDPClongNoSynd}.

\begin{proof}[Proof of Proposition~\ref{prop:qLDPClongNoSynd}]
 Correct recovery with probability $1-o(1)$ ensures that fidelity is no less than $1-o(1)$. In Theorem~\ref{thm: general Threshold} when we use $\delta_A<\frac{1}{32}$ the result in Proposition~\ref{prop:qLDPClongNoSynd} follows. 
\end{proof}

To prove the general noise threshold result in Theorem~\ref{thm: general Threshold} we generalize the proof technique from \cite{FawziGL_STOC2018} to the proposed long-range error model. To do this, we must first establish some additional definitions.\newline \newline

\begin{definition}\label{def: maxconn} $\alpha$-subset and $\text{MaxConn}_{\alpha}$ \cite{FawziGL_STOC2018}

Consider a graph $\mathcal{G}$ with a bounded degree. With respect to $\mathcal{G}$, we can define 
 $\alpha \in (0;1]$ and let 
$X, Y \subseteq \mathcal{V}$. $X$ is an $\alpha$-subset of $Y$ 
if 
$$
|X \cap Y| \geq \alpha |X|.
$$ 
We also define the integer $\text{MaxConn}_{\alpha}(Y)$ as:
$$
\text{MaxConn}_{\alpha}(Y) 
= \max \{ |X| :  \begin{aligned}[t]
   &X \text{ is connected in } \mathcal{G} \\
   &\text{and is an $\alpha$-subset of } Y\}
\end{aligned} 
$$
\end{definition}

These quantities have direct implications for the correction capabilities of quantum expander codes, as is extensively shown in \cite{FawziGL_STOC2018}. 

\begin{theorem}\label{thm: maxconn result}

The error after running the small-set-flip decoder for $(t -1)$ time-steps is represented by $E_{t}^e$ (effective error at time $t)$. Then, for $t\le \exp(c_4 n^{b_4})$, under the same conditions for $\alpha$ and $\bar{p}$ as in Theorem~\ref{thm: general Threshold},

$$
\mathbb{P} \Bigl[ \text{MaxConn}_\alpha(E_{t}^e) > c' \sqrt{n} \Bigr] \leq \exp(-c_4 n^{b_4}),
$$
for $\alpha \ln \frac{1}{\bar{p}} \ge \frac{101}{100} C_{\mc{G}}$, where  $b_4 = \epsilon_g - \gamma$, $c_4=\Omega(1)$, and $C_{\mc{G}}=\ln\left((d_{\mc{G}}-1) (1+\frac{1}{d_{\mc{G}}-2})^{d_{\mc{G}}-2}\right)$.

\end{theorem}

This probability is directly tied to the number of errors that can be successfully corrected by our model, a property that will be elaborated on in the following proof.

\begin{proof}[Proof of Theorem \ref{thm: general Threshold}]
By Proposition~3.9 in \cite{FawziGL_STOC2018}, when $MaxConn_\alpha(E_{t}^e)$ is less than the number of adversarial errors that can be corrected by the small set flip algorithm, the errors in $E_{t}^e$ are corrected by the original small set flip decoder (Alg.~2 in \cite{FawziGL_STOC2018}). 

The number of adversarial errors that can be corrected by the original small set flip decoder is at least $c' \sqrt{n}$, where $c'=\frac{r (1-8\delta_A)}{4+2r(1-8\delta_A)} \gamma_A \frac{r^2}{\sqrt{1+r^2}}$, by Proposition~4.16 in \cite{Grospellierthesis}.

By Theorem~\ref{thm: maxconn result}, $MaxConn_\alpha(E_{t}^e)$ is less than  $c' \sqrt{n}$ with probability $\ge 1 - \exp(-c_4 n^{b_4}))$ for $t \le \exp(c_4 n^{b_4})$. Hence, for $t < \exp(c_4 n^{b_4})$, the state in the memory can be retrieved using the original small set flip decoder with probability $1-o(1)$.

From \cite{FawziGL_STOC2018}, for a particular $\delta_A$, the maximum possible value of $\alpha$ is $\frac{\frac{r}{2} (1-8\delta_A)}{(1-8\delta_A)\frac{r}{2}+1}$. 

\end{proof}

For proving Theorem~\ref{thm: maxconn result} we have to introduce certain quantities similar to those in \cite{FawziGL_STOC2018}.

\begin{definition}
For a graph $\mathcal{G}$, define $\mathcal{C}_s(\mathcal{G})$ as the set of connected sets of size $s$ in $\mathcal{G}$.

\begin{align}
&\text{For } v \in V,\, X \subseteq V,\, E \subseteq V \text{ and } s \in \mathbb{N}, \text{ define:}\\
&\mathcal{C}_s(v) = \{X \in \mathcal{C}_s(\mathcal{G}) : v \in X\},\text{connected sets containing $v$,}\\
&\partial X =  \{\text{ neighbors of } X\} \setminus X,\\
&A(E, \alpha, s, v) = \left\{ X \in \mathcal{C}_s(v) : |X \cap E| \geq \alpha|X|,\, \partial X \cap E = \varnothing \right\}\\
\end{align}
\end{definition}

Next, we state two lemmas and a theorem that we subsequently use to prove Theorem~\ref{thm: maxconn result}.

\begin{lemma}\label{lem: Csv bound}
For $s=\omega(\log n)$, 
$$
|\mathcal{C}_S(v)| \leq e^{C_\mathcal{G} s} ,
$$
where $C_\mathcal{G} =  \ln\left((d_{\mc{G}}-1) (1+\frac{1}{d_{\mc{G}}-2})^{d_{\mc{G}}-2}\right)$.
\end{lemma}

\begin{definition}

\begin{align} A_{J}^t:= & \Big \{J_{i,u}, i \in [n], u \le t: \forall i \in [n] , \forall u \le t,  \nonumber \\ 
& g_i(\{J_{i,u}\}) \le \bar{p} + \frac{1}{\log n})\Big \}. 
\end{align}
\end{definition}

\begin{lemma}
\label{lem:AJ}
$\PR(A_J^t) \ge 1 - \exp(- c_4 n^{b_4})$ for $t \le \exp( c_4n^{b_4})$, where $b_4=\epsilon_g-\gamma$ and $c_4=\Omega(1)$.
\end{lemma}

\begin{theorem}\label{thm:finalFormRecursion} For the effective error at time-step $t$, given by $E_t^e$, for any $X \subset [n]$,
\begin{align*}
& \PR(\exists X \in \mc{C}_s(v): |X \cap E_{t}^e| > \alpha |X|~\big|~A_J^{t}) \\
\le& \PR(\exists X \in \mc{C}_s(v): |X \cap E_{t}| > \alpha |X|~\big|~A_J^{t})+ t \exp(-C_1 \sqrt{n}),
\end{align*}

for $\alpha \ln \frac{1}{\bar{p}} \ge \frac{101}{100} C_{\mc{G}}$ and  $C_1=\frac{c' C_{\mc{G}}}{100} $. 
\end{theorem}

\begin{proof}[Proof of Theorem~\ref{thm: maxconn result}]
Here we prove Theorem~\ref{thm: maxconn result} using the quantities defined above and Lemma~\ref{lem: Csv bound} and Theorem~\ref{thm:finalFormRecursion}.

Starting from the first equation from page 13 of \cite{FawziGL_STOC2018}, we have
$$
\mathbb{P}\Bigl[MaxConn_\alpha(E_{t}^e)>\theta\Bigr] \leq \sum\limits_{s\geq\theta}\sum\limits_{v\in V}\mathbb{P}\Bigl[A(E_{t}^e, \alpha, s, v) \neq \varnothing\Bigr]
$$

Considering  only the summand:
\begin{align*}
&\mathbb{P}\Bigl[A(E_{t}^e, \alpha, s, v) \neq \varnothing\Bigl] \nonumber \\
= & \PR\Bigl[\exists X \in \mathcal{C}_s(v): |X\cap E_{t}^e| > \alpha|X|, \partial X\cap E_{t}^e = \varnothing  \Bigl] \\
\le& \mathbb{P}\Bigl[ \exists X \in \mathcal{C}_s(v): |X\cap E_{t}^e| > \alpha|X|\Bigl]\\
 \le&\mathbb{P}\Bigl[ \exists X \in \mathcal{C}_s(v): |X\cap E_{t}^e| > \alpha|X|\Big|A_J^{t}\Bigl]+ (1-\PR(A_J^{t})) \nonumber 
\end{align*}

The last step comes from the fact that for any random variables $A$ and $B$ where $B$ is binary, $P(A) = P(B)P(A|B) + (1-P(B))P(A|\bar{B})$. Omitting the $P(B)$ and $P(A|\bar{B})$ terms gives us the required inequality, since probabilities are upper bounded by 1.\newline

Both the summed quantities in the last expression can now be bounded individually. By Lemma~\ref{lem:AJ}, we get: ~~~~
$1-\PR(A_J^{t}) \le \exp(-c_4 n^{b_4})$.\newline

Now, 
\begin{align*}
    & \mathbb{P}\Bigl[\exists X \in \mathcal{C}_s(v): |X\cap E_{t}^e| > \alpha|X|~\Big|~A_J^{t}\Bigr] \nonumber \\
     \le& \mathbb{P}\Bigl[\exists X \in \mathcal{C}_s(v): |X\cap E_{t}| > \alpha|X|~\Big|~A_J^{t}\Bigr] + t \exp(-C_1 \sqrt{n}) \\
    \le &  \sum_{X \in \mc{C}_s(v)} \mathbb{P}\Bigl[|X\cap E_{t}| > \alpha |X|~\Big|~A_J^{t}\Bigr] + t \exp(-C_1 \sqrt{n}),
\end{align*}
by invoking Theorem~\ref{thm:finalFormRecursion} and union bound. The first term can further be bounded by
\begin{align*}
     |\mc{C}_s(v)| \mathbb{P}\Bigl[|X\cap E_{t}| > \alpha |X|~\Big|~A_J^{t}\Bigr]
\end{align*}

Given any $J_{i,k} \in A_J^{t}$ for $k \le t$, $\{L_{i,k}\}$ are independent Bernoulli. Their probabilities being $1$ are upper-bounded by $\bar{p}+\frac{1}{\log n}$. By invoking simple Bernoulli coupling, the probability that $(|X \cap E_k| > \alpha |X|)$ conditioned on $J_{i,k} \in A_J^{t}$ is upper-bounded by the probability of the same event when $\{L_{i,k}\}$ are i.i.d. Bernoulli($\bar{p}+\frac{1}{\log n}$). Thus, by the Chernoff bound \cite{BoucheronLM2013}, the probability of $|X \cap E_k| > \alpha |X|$ conditioned on $J_{i,k} \in A_J^{t}$ is upper-bounded by $\exp(-|X|~\text{KL}(\alpha||\bar{p}+\frac{1}{\log n}))$. Here, KL$(a,b)$ for $a,b \in (0,1)$ is the Kullback-Leibler divergence between two Bernoulli distributions $a$ and $b$.

Thus, $\mathbb{P}\Bigl[|X\cap E_{k}| > \alpha |X|~\big|~A_J^{t}\Bigr]$ is upper-bounded by $\exp(-|X|~\text{KL}(\alpha ||\bar{p}+\frac{1}{\log n}))$. Note that for a fixed $a$ and $b$ close to $0$, the dominant term in KL$(a,b)$  is $a \ln\frac{1}{b}$.

By Lemma~\ref{lem: Csv bound}, we have $|\mc{C}_s(v)| \le \exp(s~C_\mc{G})$ for $s=\omega(\log n)$, where $C_{\mc{G}}$ is a constant that depends on the degree of $\mc{G}$ only. Hence, for $\bar{p}$ small enough such that $\alpha \ln \frac{1}{\bar{p}} \ge \frac{101}{100} C_{\mc{G}}$, $\mathbb{P}\Bigl[ A(E_{t}^e, \alpha, s, v) \neq \varnothing |~A_J^{t}\Bigl]$ vanishes exponentially as $\exp(-\frac{C_{\mc{G}}}{100} s)$.Hence,

\begin{align*}
&~~ \PR(MaxConn_{\alpha}(E_{t}^e)>c'\sqrt{n})) \nonumber \\
& \le  \sum\limits_{s\geq\theta}\sum\limits_{v\in V}\mathbb{P}\Bigl[A(E_{t}^e, \alpha, s, v) \neq \varnothing\Bigr]\nonumber \\
& \le n \cdot n \cdot (t \exp(-\frac{c' C_{\mc{G}}}{100} \sqrt{n}) + \exp(-c_4 n^{b_4})) \nonumber
\end{align*}

\end{proof}

\begin{proof}[Proof of Theorem~\ref{thm:finalFormRecursion}]
Given $A_J^t$, let us repeat some steps of the proof in Theorem~\ref{thm: maxconn result} for $t=1$ and get the bound on $\mathbb{P}\Bigl[\exists X \in \mathcal{C}_s(v): |X\cap E_{1}^e| > \alpha|X|~\Big|~A_J^{t}\Bigr]$ for $s=c' \sqrt{n}$ as $\exp(-C_1 \sqrt{n})$ for $C_1=\frac{c' C_{\mc{G}}}{100}$. Thus, with probability $1-\exp(-C_1 \sqrt{n})$, $E_2^e=E_2$ since all errors are corrected at $t=1$ if 
$MaxConn_\alpha(E_1) \le c' \sqrt{n}$. 

Repeat the same procedure at $t=2$ and so on. This, by union bound, implies that till $t$,   with probability $1 - t \cdot \exp(-C_1 \sqrt{n})$, we have $E_t^e=E_t$ given $A_J^t$.
\end{proof}

\begin{proof}[Proof of Lemma~\ref{lem:AJ}]
We first prove that 
$\PR(g_i(\{J_{i,u}\}) \le (1+\frac{1}{\log n})) \ge 1 - \exp(-c_4 n^{b_4})$. Then, the lemma follows using the union bound and the fact that $n \exp(-c_4 n^{b_4}) = \exp(-(c_4 - 1/\log n)~n^{b_4}) \approx \exp(-c_4 n^{b_4})$ for all sufficiently large $n$.

We use the result from Theorem 1 in \cite{Chazottes2007} to bound $\PR\oppar{g_i(\{J_{i,u}\}) \le \bar{p}+\frac{1}{\log n}}$, where the definition of elements of the coupling matrix D is:
\[
D_{i,j} = \underset{a,b}{max} \ \, \PR_{i,a,b}^J \oppar{J^{(1)}_j \neq J^{(2)}_j}
\]
where $\PR^J_{i,a,b}$ is the maximal coupling of the conditional laws $(J_{> i}| J_{<i}, J_i = a)$ and $(J_{> i}| J_{<i}, J_i = b)$. As $\{J_i\}$ is 1D Markov, the conditional laws become $(J_{> i}|J_i = a)$ and $(J_{>i}|J_i = b)$. Using the coupling definition of total variation distance (denoted by TV below) and the fact that $J_{i,t}$ are spatially Markov at any $t$ (not temporally Markov) we get:
\begin{align*}
    D_{i,j} = \underset{a,b}{sup} \ \, \| \PR(J_j| J_i=a) -\PR(J_j|J_i=b)\|_{TV}.
\end{align*}

This quantity has been defined before as correlation in the model and its equivalence with $\text{cov}(L_i, L_j)$ was discussed. The rest follows from the following lemma and by taking a union bound over time.
\end{proof}

\begin{lemma}
\label{lem:ChazotteWithCn}
Assuming $D_{i, j}$ decays no slower than $\frac{1}{dist(i, j)}$ and $g_i$ is $c_n$-Lipschitz, where $c_n$ is $O\left(\frac{c}{n^{0.5 + \epsilon_g}}\right)$, we have 
\begin{gather*}
\PR\Big({g_i(\{J_{i,u}\}) - \EX[g_i(\{J_i\})] \geq \frac{1}{\log n}}\Big) \leq \exp \Big({-\frac{1}{c} n^{(2 - \delta_g)\epsilon_g}}\Big)
\end{gather*}
for any $\delta_g>0$.
\end{lemma}

The Proof for the above lemma can be found in the appendix. \\
\subsection*{\bf Conclusion}
In this paper, we proved a positive noise threshold for square-root distance expander qLDPC codes against long-range correlated errors. This is the first step towards  answering the question: can constant overhead sublinear distance codes offer fault-tolerance against highly correlated errors? The analytical noise threshold  has a relatively simple expression in terms of the code rate when the code rate is small. This offers a simple tradeoff between the space overhead for these codes and the achievable noise threshold. The important practical question of local syndrome extraction for  qLDPC codes was recently answered affirmatively \cite{BerthusenGG2025,PattisonKP2025}. A future quest of practical significance would thus be to understand the performance of qLDPC codes with local syndrome extraction against long-range correlated errors.

\subsection*{{\bf Acknowledgments}}
AC thanks ANRF India (formerly, SERB India) for support through grant CRG/2023/005345 and the Ministry of Education, India.

\bibliography{bibfile}
\appendix 

\section{Proof of Lemma \ref{lem: Csv bound}}

Note that for each $v$, $|\mc{C}_s(v)| \le  |\mc{C}_s(\mc{G})|$, and from \cite{FawziGL_STOC2018} we have $|\mc{C}_s(\mc{G})| \le n \left((d_{\mc{G}}-1) (1+\frac{1}{d_{\mc{G}}-2})^{d_{\mc{G}}-2} \right)^s$. Hence, $|\mc{C}_s(v) | \le e^{C_\mc{G}s}$ for $s =\omega(\log n)$,
\\ where $C_\mc{G} = (d_{\mc{G}}-1) (1+\frac{1}{d_{\mc{G}}-2})^{d_{\mc{G}}-2}$ .

\section{Proof of Lemma~\ref{lem:ChazotteWithCn}}

From the work of Chazottes et al. \cite{Chazottes2007}, we get
\[
\mathbb{P}\left\{ g - \mathbb{E}g \geq t \right\} 
\leq \exp\left( - \frac{2t^2}{\|D\|_{\ell^2(\mathbb{N})}^2 \, \|\delta g\|_{\ell^2(\mathbb{N})}^2} \right),
\]

where $D$ is the matrix with elements $D_{i,j}$ and $\delta g$ is an $n$-dimensional vector as defined in \cite{Chazottes2007}. It follows from \cite{Chazottes2007} that if $g$ is $c_n$-Lipschitz, then each component of $\delta g$ is at most $c_n$.

Consider $V(J_i)=\Theta(n^\gamma)$ for $\gamma>0$. Note that the row-sum of $D$ is upper-bounded by $V(J_i)=\frac{n^{\gamma}}{\gamma}$. Hence, upon scaling by that factor, the matrix would be a sub-stochastic matrix whose eigenvalues are no greater than $1$. The following steps then hold:
\begin{align*}
&\PR\oppar{g_i(\{J_{i,u}\}) \le \bar{p}+\frac{1}{\log n}} \\=&\PR\oppar{g_i(\{J_{i,u}\})  - \mathbb{E}[g_i(\{J_{i,u}\}] \le \frac{1}{\log n}} \\
\leq& 
\exp\left( - \frac{2\left(\frac{1}{\log n}\right)^2}{n^{2\gamma} c_n^2n} \right)
=
\exp\left( - \frac{2\left(\frac{1}{\log n}\right)^2}{n^{2\gamma} \frac{c^2}{n^{1 + 2 \epsilon_g}} n} \right)\\
=& \exp \left( -c\frac{n^{2\epsilon_g}}{(n^{\gamma})^2 (\log n)^2} \right) \\
\le & \exp(- c n^{(2-\delta_g)(\epsilon_g-\gamma)}) \approx \exp(- c n^{2(\epsilon_g-\gamma)})
\end{align*}

The rest follows from the fact that polylogarithmic terms are smaller than any polynomial of arbitrarily small positive degree.\\
In the other case, the row-sum of $D$ is $O(1)$, and hence the bound also follows.

\end{document}